\documentclass[11pt,a4paper]{llncs}

\usepackage{fullpage}
\usepackage{graphicx}
\usepackage{amsmath}
\usepackage{amssymb}
\usepackage{algorithm}
\usepackage{algorithmic}

\newcommand{\Order}{\mathrm{O}}
\newcommand{\OrderT}{\tilde{\mathrm{O}}}

\newcommand{\poly}{\mathop{\mathrm{poly}}\nolimits}

\title{Maximum Matching in Turnstile Streams}

\author{Christian Konrad}
\institute{Reykjavik University, Reykjavik, Iceland, \email{christiank@ru.is}}

\newcommand{\Exp}{\mathbb{E}}
\renewcommand{\Pr}{\mathbb{P}}

\begin{document}
 \maketitle
\begin{abstract}
 We consider the unweighted bipartite maximum matching problem in the one-pass turnstile streaming model 
 where the input stream consists of edge insertions and deletions. In the insertion-only
 model, a one-pass $2$-approximation streaming algorithm can be easily obtained with space $\Order(n \log n)$,
 where $n$ denotes the number of vertices of the input graph. We show that no such result is possible if 
 edge deletions are allowed, even if space $\Order(n^{3/2-\delta})$ is granted, 
 for every $\delta > 0$.
 Specifically, for every $0 \le \epsilon \le 1$, we show that in the one-pass turnstile streaming model,
 in order to compute a $\Order(n^{\epsilon})$-approximation, space $\Omega(n^{3/2 - 4\epsilon})$
 is required for constant error randomized algorithms, and, 
 up to logarithmic factors, space 
 $\OrderT( n^{2-2\epsilon} )$ is sufficient.
 
 Our lower bound result is proved in the simultaneous message model of communication and may be of 
 independent interest. 
\end{abstract}

\section{Introduction}
 Massive graphs are usually dynamic objects that evolve over time in structure and size. 
 For example, the Internet graph changes as webpages are created or deleted, 
 the structure of social network graphs changes as friendships are established or
 ended, and graph databases 
 change in size when data items are inserted or deleted. 
 Dynamic graph algorithms can cope with evolving graphs of moderate sizes. They receive a 
 sequence of updates, such as edge insertions or deletions, and maintain valid solutions 
 at any moment. 
 However, when considering massive graphs, these algorithms are often less suited as they assume
 random access to the input graph, an assumption that can hardly be guaranteed in this context. 
 Consequently, research has been carried out on dynamic graph streaming algorithms that 
 can handle both edge insertions and deletions. 
 
 \textbf{Dynamic Graph Streams.}
 A data streaming algorithm processes an input stream $X = X_1, \dots, X_n$ sequentially item by item
 from left to right in passes while using a memory whose size is sublinear in the size of the input \cite{m05}. 
 Graph streams have been studied for almost two decades. However, until recently, all graph streams considered 
 in the literature were {\em insertion-only}, i.e., they process streams consisting of sequences of edge insertions. 
 In 2012, Ahn, Guha and McGregor \cite{agm12}
 initiated the study of \textit{dynamic graph streaming algorithms} that process streams consisting of both 
 edge insertions and deletions. Since then, it has been shown that a variety of problems 
 for which space-efficient streaming algorithms in the insertion-only model are known, 
 such as testing connectivity and bipartiteness, computing spanning trees, computing cut-preserving sparsifiers
 and spectral sparsifiers, can similarly be solved well in small space in the dynamic model \cite{agm12,agm132,kw14,klmms14}. 
 An exception is the maximum matching problem which, as we will detail later, is probably the most
 studied graph problem in streaming settings. In the insertion-only model, a $2$-approximation algorithm for this 
 problem can easily be obtained in one pass with $\Order(n \log n)$ space, where $n$ is the number of vertices in the
 input graph. Even in the sliding-window model\footnote{In the sliding-window model, an algorithm receives a 
 potentially infinite 
 insertion-only stream, however, only a fixed number of most recent edges are considered by the algorithm.
 Edges are seen as deleted when they are no longer contained in the most recent window of time.}, 
 which can be seen as a model located between the insertion-only model and the dynamic model, 
 the problem can be solved well \cite{cms13}.
 The status of the problem in the dynamic model has been open so far, and, in fact, the existence of sublinear 
 space one-pass dynamic streaming algorithms for the maximum matching problem
 was one of the open problems collected at the Bertinoro 2014 workshop on sublinear algorithms 
 \footnote{See also \url{http://sublinear.info/64}}.  
 
 Results on dynamic matching algorithms \cite{cdkl09,blsz14} show that 
 even when the sequence of graph updates contains deletions, then large matchings can be maintained 
 without too many reconfigurations. These results may give reasons for hope that constant 
 or poly-logarithmic approximations could be achieved in the one-pass dynamic streaming model. 
 We, however, show that if there is such an algorithm, then it uses 
 a huge amount of space. 
 
 \textbf{Summary of Our Results.} In this paper, we present a one-pass dynamic streaming algorithm for 
 maximum bipartite matching and a space lower bound for streaming algorithms in the turnstile model, 
 a slightly more general model than the dynamic model (see Section~\ref{sec:prelim} for a discussion), 
 the latter constituting the main 
 contribution  of this paper. We show that in one pass, an $\Order(n^{\epsilon})$-approximation 
 can be computed in space $\OrderT( n^{2-2\epsilon} )$ (\textbf{Theorem~\ref{thm:ub-turnstile}}), and space 
 $\Omega(n^{3/2 - 4\epsilon})$ is necessary for such an approximation (\textbf{Corollary~\ref{cor:turnstile}}).

 \textbf{Lower Bound via Communication Complexity.}
Many space lower bounds in the insertion-only model are proved in the 
 one-way communication model. In the one-way model, party one sends a message to party two who,
upon reception, sends a message to party three. This process continues until the last party receives a message 
and outputs the result. A recent result by Li, Nguy\^{e}n and Woodruff \cite{lnw14} shows that space
lower bounds for turnstile streaming algorithms can be proved
in the more restrictive {\em simultaneous model of communication} (SIM model). In this model, 
the participating parties simultaneously each send a single message to a third party, denoted the referee, 
who computes the output of the protocol as a function of the received messages. A lower bound on the size of
the largest message of the protocol is then a lower bound on the space requirements of a turnstile one-pass streaming algorithm.
Our paper is the first that uses this connection in the context of graph problems. 

A starting point for our lower bound result is a work of Goel, Kapralov and Khanna \cite{gkk12}, and a follow-up work
by Kapralov \cite{k13}. In \cite{gkk12}, via a one-way two-party communication lower bound, it is shown that in 
the insertion-only model, every algorithm that computes a $(3/2 - \epsilon)$-approximation, 
for $\epsilon > 0$, requires $\Omega(n^{1+\frac{1}{\log \log n}})$ space. This lower bound has then been
strengthened in \cite{k13} to hold for $(e/(e-1) - \epsilon)$-approximation algorithms. 
Both lower bound constructions heavily rely on 
{\em Ruzsa-Szemer\'{e}di graphs}. A graph $G$ is an $(r,s)$-Ruzsa-Szemer\'{e}di graph (in short: RS-graph), 
if its edge set can 
be partitioned into $r$ disjoint induced matchings each of size at least $s$. 
The main argument of \cite{gkk12}  can be summarized as follows: Suppose that the first party holds a relatively dense
Ruzsa-Szemer\'{e}di graph $G_1$. The second party holds a graph $G_2$ whose edges render one particular induced
matching $M \subseteq E(G_1)$ of the first party indispensable for every large matching in the graph $G_1 \cup G_2$, 
while all other induced matchings are rendered redundant. 
Note that as $M$ is an induced matching, there are no alternative edges in $G_1$ different from $M$ that 
interconnect the vertices that are matched by $M$. As the first party is not aware which of its induced matchings is required, and
as the communication budget is restricted, only few edges of $M$ on average will be sent to the second party.
Hence, the expected size of the output matching is bounded.

When implementing the previous idea in the SIM setting, the following issues have to be addressed:

 Firstly, the number of parties in the simultaneous message protocol needs to be at least as large as the
 desired bound on the approximation factor. The trivial protocol where every party sends a maximum matching 
 of its subgraph, and the referee outputs the largest received matching,
 shows that the approximation factor cannot be larger than the number of parties, even when message sizes are 
 as small as $\OrderT(n)$. Hence,
 proving hardness for polynomial approximation factors requires a polynomial number of participating parties. 
 On the other hand, the number of parties can neither be chosen too large: If the input graph is equally split 
 among $p$ parties, for a large $p$, then the subgraphs of the parties are of size $\Order(n^2 / p)$. Thus, with messages of 
 size $\OrderT(n^2 / p)$, all subgraphs can be sent to the referee who then computes and outputs an optimal solution. 
 Hence, the larger the number of parties, the weaker a bound on the message sizes can be achieved. 
 
 Secondly, there is no ``second party'' as in the one-way setting whose 
 edges could render one particular matching of every other party indispensable. Instead, a construction is required so that 
 every party both has the function of party one (one of its induced matchings is indispensable for every large
 matching) and of party two 
 (some of its edges render many of the induced matchings of other parties redundant). This suggests 
 that the RS-graphs of the parties have to overlap in many vertices. While 
 arbitrary RS-graphs with good properties can be employed for the lower bounds of \cite{gkk12} and \cite{k13}, 
 we need RS-graphs with simple structure in order to coordinate the overlaps between the parties.

We show that both concerns can be handled. In Section~\ref{sec:input-dist}, we present a carefully designed input distribution where each party holds a highly symmetrical RS-graph. The RS-graph of a party overlaps
almost everywhere with the RS-graphs of other parties, except in one small induced matching. This matching,
however, cannot be distinguished by the party, and hence, as in the one-way setting, the referee will not 
receive many edges of this matching. 

\textbf{Upper Bound.} 
 Our upper bound result is achieved by an implementation of a simple matching algorithm in the dynamic 
 streaming model: 
 For an integer $k$, pick a random subset $A' \subseteq A$ of size $k$ of one bipartition of the bipartite 
 input graph $G=(A, B, E)$; for each 
 $a \in A'$, store arbitrary $\min\{k, \deg(a) \}$ incident edges, where $\deg(a)$ denotes the degree
 of $a$ in the input graph; output a maximum matching in the
 graph induced by the stored edges. We prove that this algorithm has an approximation factor of $n/k$.
 In order to collect $k$ incident edges of a given vertex in the dynamic streaming model,
 we employ the $l_0$-samplers of Jowhari, Sa\u{g}lam, Tardos \cite{jst11}, which have previously been used for 
 dynamic graph streaming algorithms \cite{agm12}. By chosing $k = \Theta(n^{1-\epsilon})$, this construction leads to a $\Order(n^{\epsilon})$-approximation
 algorithm with space $\OrderT(n^{2-2\epsilon})$. 
 While this algorithm in itself is rather simple and standard, it shows that non-trivial approximation 
 ratios for maximum bipartite matching in the dynamic streaming model are possible with sublinear space. 
 Our upper and lower bounds show that in order to compute a $n^\epsilon$-approximation, space $\OrderT(n^{2-2\epsilon})$
 is sufficient and space $\Omega(n^{3/2 - 4\epsilon})$ is required. Improving on either side 
 is left as an open problem.
 
 \textbf{Further Related Work.}
 Matching problems are probably the most studied graph problem in the streaming model 
 \cite{fkmsz05,mc05,eks09,elms10c,ag11,agm12,kmm12,z12,gkk12,k13,go13,cms13,cs14,klmms14,m14,kks14,ehlmo15}.
 Closest to our work are the already mentioned lower bounds \cite{gkk12} and \cite{k13}. Their arguments
 are combinatorial and so are the arguments in this paper. Note that lower bounds for 
 matching problems in communication settings have also been obtained via information complexity in 
 \cite{go13,hrvz15}.
 
  In the dynamic streaming model, Ahn, Guha, and McGregor \cite{agm12} provide a multi-pass algorithm
 with $\Order(n^{1+1/p} \poly \epsilon^{-1})$ space, $\Order(p \cdot \epsilon^{-2} \cdot \log \epsilon^{-1})$
 passes, and approximation factor $1+\epsilon$ for the weighted maximum matching problem, for a parameter $p$. 
 This is the only result on matchings known in the dynamic streaming setting.

\textbf{Recent Related Work.} Assadi et. al. \cite{akly15} independently and concurrently to this work essentially 
resolve the questions asked in this paper. Using the same techniques ($l_0$-sampling for the upper bound,
simultaneous communication complexity and Rusza-Szemer\'{e}di graphs for the lower bound), they show that there is a 
$\Order(n^\epsilon)$-approximation dynamic streaming 
algorithm for maximum matching which uses 
$\Order(n^{2-3\epsilon})$ space. 
Furthermore, they prove that this is essentially tight for turnstile
algorithms: Any such algorithm in the turnstile model requires space at least $n^{2-3\epsilon - o(1)}$. 
 
\textbf{Outline.} We start our presentation with a section on preliminaries. Then, in Section~\ref{sec:input-dist}, we 
present our hard input distribution which is then used in Section~\ref{sec:lb} in order to prove our 
lower bound in the SIM model. Finally, we conclude with
our upper bound in Section~\ref{sec:ub}.
 
\vspace{-0.3cm}
\section{Preliminaries} \label{sec:prelim}

\vspace{-0.2cm}
For an integer $a \ge 1$, we write $[a]$ for $\{1, \dots, a \}$. 
We use the notation $\OrderT()$, which equals the standard $\Order()$ notation
where all poly-logarithmic factors are ignored. 


\noindent \textbf{Simultaneous Communication Complexity.} Let $G = (A, B, E)$ denote a simple bipartite graph,
and, for an integer $P \ge 2$, let $G_1, \dots, G_P$ be edge-disjoint subgraphs of $G$. In the simultaneous 
message complexity setting, for $p \in [P]$, party $p$ is given $G_p$, and 
sends a single message $\mu_p$ of limited size to a third party denoted the referee. Upon 
reception of all messages, the referee outputs a matching $M$ in $G$. Note that the participating parties cannot
communicate with each other, but they have access to an infinite number of shared random coin flips which can
be used to synchronize their messages.


We say that an algorithm/protocol is a constant error algorithm/protocol if it
errs with probability at most $\epsilon$, for $0 \le \epsilon < 1/2$. We also assume that
a algorithm/protocol never outputs edges that do not exist in the input graph. 

\noindent \textbf{Turnstile streams.} For a bipartite graph $G=(A,B,E)$, 
let $X = X_1, X_2, \dots$ be the
input stream with $X_i \in E \times \{+1, -1\}$, where $+1$ indicates that an edge is inserted, and $-1$ 
indicates that an edge is deleted. Edges could potentially be inserted multiple times, or be removed before
they have been inserted, as long as once the stream has been fully processed, the multiplicity of an edge is 
in $\{-c, -c +1, \dots, c-1, c\}$, for some integer $c$. The reduction of \cite{lnw14} and hence our 
lower bound holds for algorithms that can handle this type of dynamic streams, also known as \textit{turnstile streams}.
Such algorithms may for instance abort if negative edge multiplicities are encountered, or they 
output a solution among the edges with non-zero multiplicity. 

In \cite{lnw14} it is shown that every turnstile algorithm can be seen as an algorithm that solely computes
a linear sketch of the input stream. As linear sketches can be implemented in the SIM model,
lower bounds in the SIM model are lower bounds on the sketching complexity of problems, which in turn
imply lower bounds for turnstile algorithms.
We stress that our lower bound holds for linear sketches. Note that {\em all} known dynamic 
graph algorithms\footnote{Some of those algorithms couldn't handle arbitrary turnstile streams as they 
rely on the fact that all edge multiplicities are in $\{0, 1\}$.} 
solely compute linear sketches (e.g. \cite{agm12,agm132,kw14,klmms14}). 
This gives reasons to conjecture that also all dynamic algorithms can be seen as linear
sketches, and, as a consequence, our lower bound not only holds for turnstile algorithms but for all dynamic algorithms.


\vspace{-0.5cm}

. 



\vspace{-0.0cm}
\section{Hard Input Distribution} \label{sec:input-dist}

\vspace{-0.0cm}
In this section, we construct our hard input distribution. First, we describe the construction of the distribution 
from a global point of view in Subsection~
3.1. Restricted to the input graph $G_p$ of
any party $p \in [P]$, the distribution of $G_p$ can be described by a different construction which is simpler and 
more suitable for our purposes. This will be discussed in Subsection~
3.2.

\textbf{3.1. Hard Input Distribution: Global View}
Denote by $P$ the number of parties of the simultaneous message protocol. 
Let $k,Q$ be integers so that $P \le k \le \frac{n}{P}$, and $Q = o(P)$. The precise values of $k$ and $Q$ will 
be determined later. 
First, we define a bipartite graph $G' = (A, B, E)$ on $\Order(n)$ vertices
with $A = B = [(Q+P)k]$ from which we obtain our hard input distribution. 
For $1 \le i \le Q+P$, let $A_i = [1 + (i-1)k, ik]$ and let $B_i = [1 + (i-1)k, ik]$. 
The edge set $E$ is a collection of matchings as follows:
\begin{eqnarray*}
 E = \bigcup_{i,j \in [Q], p \in [P]} M_{i,j}^p \, \cup \bigcup_{i \in \{Q+1, \dots, Q+P\}, j \in [Q]} (M_{i,j} \cup M_{j,i}) \, \cup \bigcup_{i \in \{Q+1, \dots, Q+P\}} M_{i,i},
\end{eqnarray*}
where $M_{i,j}$ is a perfect matching
between $A_i$ and $B_j$, and $M_{i,j}^1, \dots, M_{i,j}^P$ are $P$
edge-disjoint perfect matchings between $A_i$ and $B_j$. Note that as we required that $k \ge P$, 
the edge-disjoint matchings $M_{i,j}^1, \dots, M_{i,j}^P$ can be constructed\footnote{For instance, define $G'$ 
so that $G'|_{A_i \cup B_i}$ is a $P$-regular bipartite graph. It is well-known (and easy to see 
via Hall's theorem) that any $P$-regular bipartite graph is the union of $P$ edge-disjoint perfect 
matchings.}.

From $G'$, we construct the input graphs of the different parties as follows:
\begin{enumerate}
 \item For every $p \in [P]$, let $G_p' = (A, B, E_p')$ where $E_p'$ consists of the matchings $M_{i,j}^p$ 
 for $i,j \in [Q]$, the matching $M_{Q+p, Q+p}$ and the matchings $M_{Q+p,j}$ and $M_{j,Q+p}$ for $j \in [Q]$.
 
 \item For every $p \in [P]$, for every matching $M$ 
 of $G_p'$, pick a subset of edges of size $k/2$ from $M$ uniformly at random and replace $M$ by this subset.
 
 \item Pick random permutations $\pi_A, \pi_B: [Q+P] \rightarrow [Q+P]$. 
 Permute the vertex IDs of the graphs $G_p'$,
 for $1 \le p \le P$, so that if $\pi_A(i) = j$ then $A_i$ receives the IDs of $A_j$ as follows: The vertices
 $a_1 = 1+ k (i-1), a_2 = 2+ k(i-1), \dots, a_k = ki$ receive new IDs so that after the change of IDs, we have
 $a_1 = 1+ k (j-1), a_2 = 2+ k(j-1), \dots, a_k = kj$. The same procedure is carried out with vertices $B_i$ and
 permutation $\pi_B$. Denote by $G_p$ the graph $G_p'$ once half of the edges have been removed and the vertex IDs
 have been permuted. Let $G$ be the union of the graphs $G_p$.
\end{enumerate}
\vspace{-0.1cm}
The structure of $G'$ and a subgraph $G_p'$ is illustrated in Figure~\ref{fig:structure-hi}.
\vspace{-0.4cm}
\begin{figure}
\begin{center}
 \includegraphics[height=2cm]{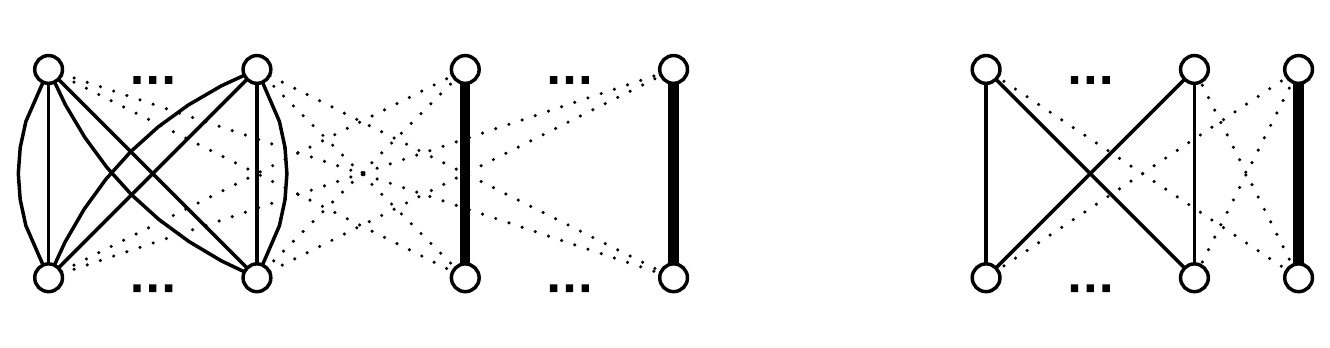} 
 
 \vspace{-2.33cm}
 
 $\, \quad\, A_1 \quad\quad \, \, A_Q$ \hspace{0.5cm} $A_{Q+1} \quad A_{Q+P}$   \hspace{0.8cm} $\, \, A_1 \quad \quad \, \, A_Q \, \, \,  A_{Q+p}$
 
 \vspace{1.7cm}
 
 $\, \quad\, B_1 \quad\quad \, \, B_Q$ \hspace{0.5cm} $B_{Q+1} \quad B_{Q+P}$   \hspace{0.8cm} $\, \, B_1 \quad \quad \, \, B_Q \, \, \,  B_{Q+p}$
 
 \vspace{-1.4cm} $\quad \quad \quad G'$ \hspace{8cm} $G_p' \subseteq G_p$
\end{center}

\vspace{0.55cm} 

\caption{Left: Graph $G'$. A vertex corresponds to a group of $k$ vertices. Each edge indicates a perfect 
matching between the respective vertex groups. The bold edges correspond to the matchings $M_{Q+p,Q+p}$, for
$1 \le p \le P$, the solid edges correspond to matchings $M_{i,j}^p$, for $1\le i,j \le Q$, $1 \le p \le P$,
and the dotted edges correspond to matchings $M_{Q+p, i}, M_{i, Q+p}$, for $1 \le i \le Q$ and $1 \le p \le P$. 
Right: Subgraph $G_p' \subseteq G$. \label{fig:structure-hi}}
\end{figure}
\vspace{-0.4cm}

\textit{Properties of the input graphs.} Graph $G'$ has a perfect matching of size $(Q+P)k$ which consists of
a perfect matching between vertices $A_1, \dots, A_Q$ and $B_1, \dots, B_Q$, and the matchings $M_{Q+p, Q+p}$
for $1 \le p \le P$. As by Step~2 of the construction of the hard instances, we remove half of the edges
of every matching, a maximum matching in graph $G$ is of size at least $\frac{(Q+P)k}{2}$.
Note that while there are many possibilities to match the vertex groups $A_1, \dots, A_Q$ and $B_1, \dots, B_Q$,
in every large matching, many vertices of $A_{Q+i}$ are matched to vertices $B_{Q+i}$ using edges from
the matching $M_{Q+i, Q+i}$. 
For some $p \in [P]$, consider now the graph $G_p'$ from which the graph $G_p$ is 
constructed. $G_p'$ consists of perfect matchings between the vertex groups $A_i$ and $B_j$ for every
$i,j \in [Q] \cup \{p \}$. In graph $G_p$, besides the fact that only half of the edges of every matching are kept, 
the vertex IDs are permuted. We will argue that due to the permuted vertices, given $G_p$, it is difficult to 
determine which of the matchings corresponds to the matching $M_{Q+p,Q+p}$ in $G'$. Therefore, if the referee
is able to output edges from the matching $M_{Q+p,Q+p}$, then many edges from every matching
have to be included into the message $\mu_p$ sent by party $p$.

\textbf{3.2. Hard Input Distribution: Local View.} 
From the perspective of an individual party, by symmetry of the previous construction, the distribution 
from which the graph $G_p$ is chosen can also be described as follows: 
\vspace{-0.1cm}
\begin{enumerate}
 \item Pick $I_A, I_B \subseteq [Q+P]$ so that $|I_A| = |I_B| = Q+1$ uniformly at random. 
 \item For every $i \in I_A$ and $j \in I_B$, introduce a matching of size $k/2$ between $A_i$ and $B_j$ chosen uniformly
 at random from all possible matching between $A_i$ and $B_j$ of size $k/2$.
\end{enumerate}
$G_p$ can be seen as a $((Q+1)^2 , k/2)$-Ruzsa-Szemer\'{e}di graph or as a $(Q+1,k(Q+1)/2)$-Ruzsa-Szemer\'{e}di graph.
Let $\mathcal{G}_p$ denote the possible input graphs of party $p$. We prove now a lower bound on $|\mathcal{G}_p|$.
\begin{lemma} \label{lem:nbr-input-graphs}
 There are at least
  $|\mathcal{G}_p| > {Q+P \choose Q+1} \frac{(Q+P)!}{(P-1)!} \left( \frac{2^k}{k+1} \right)^{(Q+1)^2}$ 
possible input graphs for every party $p$. Moreover, the input distribution is uniform.
\end{lemma}
\begin{proof}
The vertex groups $I_A$ and $I_B$ are each of cardinality $Q+1$ and chosen from the set $[Q+P]$. There are 
${Q+P \choose Q+1}$ choices for $I_A$. Consider one particular choice of $I_A$. Then, there
are $\frac{(Q+P)!}{(P-1)!}$ possibilities to pair those with $Q+1$ vertex groups of the $B$ nodes.
Each matching is a subset of $k/2$ edges from $k$ potential edges. Hence, there are 
${Q+P \choose Q+1} \frac{(Q+P)!}{(P-1)!} {k \choose \frac{1}{2} k}^{(Q+1)^2}$ input graphs
for each party.
 Using a bound on the central binomial coefficient, this term can be bounded from below by
 ${Q+P \choose Q+1} \frac{(Q+P)!}{(P-1)!} \left( \frac{2^k}{k+1} \right)^{(Q+1)^2}$. 
\qed
\end{proof}
The matching in $G_p$ that corresponds to the matching between $A_{Q+p}$ and $B_{Q+p}$ in $G_p'$ will 
play an important role in our argument. In the previous construction, every introduced matching in $G_p$ plays
the role of matching $M_{Q+p, Q+p}$ in $G_p'$ with equal probability. In the following, we will denote
by $M_p$ the matching in $G_p$ that corresponds to the matching $M_{Q+p, Q+p}$ in $G_p'$.



\vspace{-0.0cm}
\section{Simultaneous Message Complexity Lower Bound} \label{sec:lb}

\vspace{-0.0cm}
We prove now that no communication protocol with limited maximal message size 
performs well on the input distribution described in Section~\ref{sec:input-dist}. 
First, we focus on deterministic protocols, and we prove a lower bound on the expected 
approximation ratio (over all possible input graphs) of any deterministic protocol (Theorem~\ref{thm:det-lb}). Then, 
via an application of Yao's lemma, we obtain our result for randomized constant 
error protocols (Theorem~\ref{thm:rand-lb}). Our lower bound for dynamic one-pass streaming algorithms,
Corollary~\ref{cor:turnstile}, is then obtained as a corollary of 
Theorem~\ref{thm:rand-lb} and the reduction of \cite{lnw14}. 

\textit{Lower Bound For Deterministic Protocols.} 
Consider a deterministic protocol that runs on a hard instance graph $G$ and uses messages of length at
most $s$. 
As the protocol is deterministic, for every party 
$p \in [P]$, there exists a function $m_p$ that maps the input graph $G_p$ of party $p$ to a message $\mu_p$. 
As the maximum message length is limited by $s$, there are $2^s$ different possible messages. Our parameters 
$Q, k$ will be chosen so that $s$ is much smaller than the number of input graphs $G_p$ for party $p$, as 
stated in Lemma~\ref{lem:nbr-input-graphs}. Consequently, many input graphs are mapped to the same message. 

Consider now a message $\mu_p$ and denote by $\mu_p^{-1}$ the set of graphs $G_p$ that are mapped by 
$m_p$ to message $\mu_p$. Upon reception of $\mu_p$, the referee can only output edges that are contained 
in \textit{every} graph of $\mu_p^{-1}$, since all outputted edges have to be contained in the input graph. 

Let $N$ denote the matching outputted by the referee, and let $N_p = N \cap M_p$ denote the outputted edges from 
matching $M_p$. Furthermore, for a given message $\mu_p$, denote by $G_{\mu_p} := M_p \cap \bigcap_{G_p \in \mu_p^{-1}} G_p$.

In the following, we will bound the quantity $\Exp |N_p|$ from above 
(Lemma~\ref{lem:one-party-matching-size}). 
By linearity of expectation, this allows us to argue about the expected number of edges of the matchings 
$\cup_p N_p$ outputted by the referee. We can hence argue about the expected size of the outputted 
matching, which in turn implies a lower bound on the approximation guarantee of the protocol 
(Theorem~\ref{thm:det-lb}).

\begin{lemma} \label{lem:one-party-matching-size}
For every party $p \in [P]$, we have
$\Exp |N_p| = \Order \left( \frac{\sqrt{sk}}{Q} \right).$
\end{lemma}
\begin{proof}
Let $\Gamma$ denote the set of potential messages from party $p$ to the referee. 
As the maximum message length is bounded by $s$, we have $|\Gamma| \le 2^s$.
Let $V = \frac{|\mathcal{G}_p|}{k 2^s}$ be a parameter which splits the set $\Gamma$ into two 
parts as follows. Denote by $\Gamma_{\ge} \subseteq \Gamma$ the 
set of messages $\mu_p$ so that $|\mu_p^{-1}| \ge V$, and let $\Gamma_{<} = \Gamma \setminus \Gamma_{\ge}$.
In the following, for a message $\mu_p \in \Gamma$, we denote by $\Pr \left[ \mu_p \right]$ the
probability that message $\mu_p$ is sent by party $p$. 
Note that 
$\sum_{\mu_p \in \Gamma_{<}}  \Pr \left[ \mu_p \right] < \frac{2^{s} V}{|\mathcal{G}_p|}$, since there
are at most $2^{s} V$ input graphs that are mapped to messages in $\Gamma_{<}$. We hence obtain:
\begin{eqnarray*}
 \Exp |N_p| & \le & 
 \sum_{\mu_p \in \Gamma} \Pr \left[ \mu_p \right] \Exp |G_{\mu_p}| 
 = \sum_{\mu_p \in \Gamma_{\ge}} \left( \Pr \left[ \mu_p \right] \Exp |G_{\mu_p}| \right) + 
 \sum_{\mu_p \in \Gamma_{<}} \left( \Pr \left[ \mu_p \right] \Exp |G_{\mu_p}| \right) \\
 & \le & \sum_{\mu_p \in \Gamma_{\ge} } \left( \frac{|\mu_p^{-1}|}{2^s} \Exp |G_{\mu_p}| \right) + 
 \sum_{\mu_p \in \Gamma_{<}} \left( \Pr \left[ \mu_p \right] \right) k \\
 & < & \max \{ \Exp |G_{\mu_p}| \, : \, \mu_p \in \Gamma_{\ge} \} + 
 \frac{2^{s} V}{|\mathcal{G}_p|} k = \max \{ \Exp |G_{\mu_p}| \, : \, \mu_p \in \Gamma_{\ge} \} + 1,
\end{eqnarray*}
where we used the definition of $V$ for the last equality.
In Lemma~\ref{lem:one-message}, we prove that $\forall \mu_p \in \Gamma_{\ge}: \Exp |G_{\mu_p}| = 
\Order\left( \frac{\sqrt{sk}}{Q} \right)$.  This then implies the result.
\qed
\end{proof}

\begin{lemma} \label{lem:one-message}
 Suppose $\mu_p$ is so that $|\mu_p^{-1}| \ge V = \frac{|\mathcal{G}_p|}{k 2^s}$. Then,
  $\Exp |G_{\mu_p}| = \Order \left(\frac{\sqrt{sk}}{Q} \right).$
\end{lemma}
\begin{proof}
Remember that every graph $G_p \in \mu_p^{-1}$ consists of $(Q+1)^2$ edge-disjoint matchings,
and $M_p$ is a randomly chosen one of those. We define
\begin{eqnarray*}
 I_l = \{ (i,j) \in [Q+P] \times [Q+P] \, : \, G_{\mu_p}|_{A_i \cup B_j} \mbox{ contains a matching of size $l$} \}.
 \end{eqnarray*}
 We prove first that if $|I_l|$ is large, then $\mu_p^{-1}$ is small.
 \vspace{-0.2cm}
 \begin{claim} \label{claim} Let $l = o(k)$. Then,
 $|I_l| \ge x \Rightarrow |\mu_p^{-1}| < {Q+P \choose Q+1} \frac{(Q+P)!}{(P-1)!} (\frac{3}{4})^{lx} \left( \frac{2^k}{\sqrt{k}} \right)^{(Q+1)^2}.$
 \end{claim}
 \begin{proof}
Every graph of $\mu_p^{-1}$ contains $l$ edges of $x$ (fixed) matchings. The remaining
edges and remaining matchings can be arbitrarily chosen. Then, by a similar argument as in the proof of 
Lemma~\ref{lem:nbr-input-graphs}, we obtain
 \begin{eqnarray*}
  |\mu_p^{-1}| & \le & {Q+P-x \choose Q+1-x} \frac{(Q+P-x)!}{(P-1)!} {k-l \choose \frac{1}{2}k - l}^x {k \choose \frac{1}{2} k}^{(Q+1)^2 - x } \\
  & <  &
  {Q+P \choose Q+1} \frac{(Q+P)!}{(P-1)!} (\frac{3}{4})^{lx} \left( \frac{2^k}{\sqrt{k}} \right)^{(Q+1)^2},
 \end{eqnarray*}
 where we used ${k-l \choose \frac{1}{2}k - l} = {k-l \choose \frac{1}{2}k}$, and the bound
 ${k-l \choose \frac{1}{2}k} < (\frac{3}{4})^l {k \choose \frac{1}{2}k}$ (remember: $l = o(k)$) which follows
 from Lemma~\ref{lem:bin-bound} (see Appendix). \qed
\end{proof}

 Then, we can bound:
 \begin{eqnarray}
  \Exp |G_{\mu_p}| \le \frac{|I_l|}{ { Q+1 \choose 2} } \cdot k + (1 - \frac{|I_l|}{ { Q+1 \choose 2} }) l < \frac{|I_l|}{ { Q+1 \choose 2} } \cdot k + l. \label{eqn:291}
 \end{eqnarray}
Note that by assumption, we have $\mu_p^{-1} \ge V$. Let $l,x$ be two integers so
that: 
\begin{eqnarray}
{Q+P \choose Q+1} \frac{(Q+P)!}{(P-1)!} (\frac{3}{4})^{lx} \left( \frac{2^k}{\sqrt{k}} \right)^{(Q+1)^2} = V.  \label{eqn:397}
\end{eqnarray}
 Then, by the previous claim, we obtain $|I_l| < x$. Solving Equality~\ref{eqn:397} for variable $x$, and
 further bounding it yields:
 \begin{eqnarray}
 x& \le & \frac{1}{l} \left( (Q+1)^2(k-\frac{1}{2} \log k) + \log \left( {Q+P \choose Q+1} \frac{(Q+P)!}{(P-1)!} \right) - \log V \right) . \label{eqn:110}
 \end{eqnarray}
Remember that $V$ was chosen as $V = \frac{|\mathcal{G}_p|}{k2^s}$, and hence 
$\log V \ge (Q+1)^2(k - \log (k+1)) + \log \left( {Q+P \choose Q+1} \frac{(Q+P)!}{(P-1)!} \right) - s - \log(k)$. Using this bound 
in Inequality~\ref{eqn:110} yields
\begin{eqnarray*}
 x & \le & \frac{1}{l} \left( (Q+1)^2 ( \log(k+1) - \frac{1}{2} \log k) + s - \log k \right) \\
 & \le & \frac{1}{l} \left( (Q+1)^2 ( \log(k+1)) + s \right).
\end{eqnarray*}
 Now, using $|I_l| \le x$ and the previous inequality on $x$, we continue simplifying Inequality~\ref{eqn:291}
 as follows:
\begin{eqnarray*}
   \Exp |G_{\mu}| & \le & \dots < \frac{|I_l|}{ (Q+1)^2 } \cdot k + l \le 
   \frac{(Q+1)^2 ( \log(k+1)) + s }{l (Q+1)^2 } \cdot k + l    \\
   & \le & \frac{\log(k+1) k}{l} + \frac{sk}{l (Q+1)^2} + l = \Order( \frac{sk}{l (Q+1)^2} + l ),
\end{eqnarray*}
since $s = \omega( (Q+1)^2 \log (k+1))$.
We optimize by choosing $l = \frac{\sqrt{sk}}{Q}$, and we conclude $\Exp |G_{\mu}| = \Order(\frac{\sqrt{sk}}{Q})$.
\qed 
\end{proof}

\begin{theorem} \label{thm:det-lb}
 For any $P \le \sqrt{n}$, let $\mathcal{P}_{\mbox{det}}$ be a $P$-party deterministic simultaneous message 
 protocol for maximum matching where all messages are of size at most $s$. Then, $\mathcal{P}_{\mbox{det}}$ 
 has an expected approximation factor of $\Omega \left( \left( \frac{Pn}{s} \right)^{\frac{1}{4}} \right)$. 
\end{theorem}

\begin{proof}
 For every matching $M'$ in the input graph $G$, the size of $M'$ can be bounded by 
 $|M'| \le 2 Qk + \sum_{p=1}^{P} |M' \cap M_p|$, since at most $2Qk$ edges can be matched 
 to the vertices of the vertex groups $\bigcup_{i \in [Q]} A_i \cup B_i$, and the edges of matchings $M_p$
 are the only ones not incident to any vertex in $\bigcup_{i \in [Q]} A_i \cup B_i$. Hence, 
 by linearity of expectation, and the application of Lemma~\ref{lem:one-party-matching-size}, we obtain:
 \begin{eqnarray}
  \Exp |N| \le 2Qk + \sum_{p=1}^P \Exp |N_p| \le 2 Qk + P \cdot \Order \left( \frac{\sqrt{sk}}{Q} \right). \label{eqn:303}
 \end{eqnarray}
 A maximum matching in $G$ is of size at least $\frac{k(Q+P)}{2}$. We hence obtain the expected approximation factor:  
 \begin{eqnarray}
\nonumber  \Exp \frac{\frac{1}{2} k(Q+P)}{|N|} & \ge & \frac{\frac{1}{2} k(Q+P)}{\Exp |N|} = \Omega \left( \frac{k(Q+P)}{  \left( Qk + P \cdot \frac{\sqrt{sk}}{Q} \right)} \right) = \Omega \left( \frac{(Q+P)Q \sqrt{k}}{Q^2\sqrt{k} + P\sqrt{s} } \right)\\
  & = &  \Omega \left( \frac{P Q \sqrt{k}}{Q^2\sqrt{k} + P\sqrt{s} } \right) =  \Omega \left( \min \{ \frac{P}{Q} , \frac{Q\sqrt{k}}{\sqrt{s}} \} \right), \label{eqn:005}
 \end{eqnarray} 
 where the first inequality follows from Jensen's inequality, and the third equality uses $Q = o(P)$. The previous expression is maximized for
 $Q = \left( \frac{P \sqrt{s}}{\sqrt{k}} \right)^{1/2}$, and we obtain an approximation factor of 	
 $\Omega \left( \frac{P^{\frac{1}{2}} k^{\frac{1}{4}}}{s^{\frac{1}{4}}} \right)$.
In turn, this expression is maximized when $k$ is as large as possible, that is, $k = n/P$ (remember that
the possible range for $k$ is $P \le k \le n/P$). We hence conclude
that the approximation factor is $\Omega(\left( \frac{Pn}{s} \right)^{\frac{1}{4}})$. 
 \qed
\end{proof}


\textit{Lower Bound for Randomized Protocols.}
Last, in Theorem~\ref{thm:rand-lb} (proof in appendix), we extend our determinstic lower bound 
to randomized ones. 

\begin{theorem} \label{thm:rand-lb}
 For any $P \le \sqrt{n}$, let $\mathcal{P}_{\mbox{rand}}$ be a $P$-party randomized simultaneous message 
 protocol for maximum matching with error at most $\epsilon < 1/2$, and all messages are of size at most $s$. Then, $\mathcal{P}_{\mbox{rand}}$ has an approximation factor of $\Omega \left( \left( \frac{Pn}{s} \right)^{\frac{1}{4}} \right)$. 
\end{theorem}

Our lower bound for one-pass turnstile algorithms now follows from the reduction given in \cite{lnw14} and the application
of Theorem~\ref{thm:rand-lb} for $P = \sqrt{n}$.

\begin{corollary} \label{cor:turnstile}
 For every $0 \le \epsilon \le 1$, every randomized constant error turnstile one-pass streaming algorithm 
 for maximum bipartite matching with approximation ratio $n^{\epsilon}$ uses space $\Omega \left(n^{\frac{3}{2} - 4 \epsilon} \right)$.
\end{corollary}

\vspace{-0.3cm}
\section{Upper Bound} \label{sec:ub}
\vspace{-0.3cm}
\begin{algorithm}[H] \caption{Bipartite Matching algorithm \label{alg:ub}}
\begin{algorithmic}[1]
 \REQUIRE $G = (A, B, E)$ \{Bipartite input graph\}
 \STATE $A' \gets $ subset of $A$ of size $k$ chosen uniformly at random
 \STATE $\forall a \in A': $ $E'[a] \gets $ arbitrary subset of incident edges of $a$ of size $\min\{k, \deg_G(a) \}$ \label{line:900}
 \RETURN maximum matching in $\bigcup_{a \in A'} E'[a]$
\end{algorithmic}
\end{algorithm}

\vspace{-0.3cm}

In this section, we first present a simple randomized algorithm for bipartite matching. 
Then, we will discuss implementations of this algorithm as 
a simultaneous message protocol and as a dynamic one-pass streaming algorithm.



\textit{Bipartite Matching Algorithm.}
Consider Algorithm~\ref{alg:ub}. First, a subset $A' \subseteq A$ consisting
of $k$ vertices is chosen uniformly at random. Then, for each vertex $a \in A'$, the algorithm picks arbitrary 
$k$ incident edges. Finally, a maximum matching among the retained edges is computed and returned.

Clearly, the algorithm stores at most $k^2$ edges. The proof of the next lemma concerning the approximation
ratio of Algorithm~\ref{alg:ub} is deferred to the appendix.

\vspace{-0.01cm}
\begin{lemma} \label{lem:ub}
 Let $G = (A, B, E)$ be a bipartite graph with $|A| + |B| = n$. Then, Algorithm~\ref{alg:ub} has an 
 expected approximation ratio of $\frac{n}{k}$. 
\end{lemma}
\noindent \textit{Notations.} In the proof of Lemma~\ref{lem:ub}, we use the following additional notation.
Let $G = (A, B, E)$ be a bipartite graph. For a set of edges $E' \subseteq E$, we denote by $A(E')$ 
the subset of $A$ vertices $a$ for which there exists at least one edge in $E'$ incident to $a$. The set $B(E')$
is defined similarly. 
\begin{proof}
 Let $M$ denote the output of the algorithm, let $M^*$ be a maximum matching in $G$, and
 let $E' = \bigcup_{a \in A'} E'[a]$.
 Let $A'^* = A' \cap A(M^*)$. As $A'$ has been chosen uniformly at random, we have 
 $\Exp |A'^*| = \frac{k |A(M^*)|}{|A|}$. We will prove now that 
 the algorithm can match all vertices in $A'^*$. This then implies the result, as $|M| \ge |A'^*|$, 
 and 
 \begin{eqnarray*}
\Exp \frac{|M^*|}{|M|} \le \Exp \frac{|M^*|}{|A'^*|} = \frac{|M^*| |A|}{k |A(M^*)|} = \frac{|A|}{k} \le n/k,  
 \end{eqnarray*}
 where we used $|A(M^*)| = |M^*|$.

 To this end, we construct a matching $M'$ that matches all vertices of $A'^*$. Let $A'^*_1 \subseteq 
 A'^*$ so that for every $a \in A'^*_1$, the incident edge of $a$ in $M^*$ has been retained by the 
 algorithm. Denote by $M_1 \subseteq M^*$ the subset of optimal edges incident to the vertices $A'^*_1$. 
 Then, let $A'^*_2 = A'^* \setminus A'^*_1$.

  Consider now the graph $\tilde{G}$ on vertices $A'^*_2$ and $B \setminus B(M_1)$ and edges
 $$\{e \in E' \, : \, e = (a,b) \mbox{ with } a \in A'^*_2 \mbox{ and } b \in B \setminus B(M_1) \}.$$
 Note that as for every vertex $a \in A'^*_2$, its optimal incident edge has not been retained, $k$ different
 edges have been retained (which also implies that the degree of $a$ in $G$ is at least $k$). 
 Therefore, the degree of every $a \in \tilde{G}$ is at least $k - |B(M_1)| = k - |M_1|$. Furthermore, 
 note that $|A'^*_2| = k - |A'^*_1| = k - |M_1|$.  Thus, by Hall's marriage theorem, there exists a matching 
 $M_2$ in $\tilde{G}$ matching all vertices $A'^*_2$, and hence, $|M_2| = k - |M_1|$.
 
 We set $M' = M_1 \cup M_2$ and all vertices of $A'^*$ are matched. We obtain $|M'| = |M_1| + |M_2| = k$, and 
 the result follows. \qed
\end{proof}

\noindent \textit{Implementation of Algorithm~\ref{alg:ub} as a Simultaneous Message Protocol.} 
Algorithm~\ref{alg:ub} can be implemented in the simultaneous message model as follows. 
Using shared random coins, the $P$ parties agree on the subset $A' \subseteq A$. Then, 
for every $a \in A'$, every party chooses arbitrary $\min \{ \deg_{G_i}(a), k \}$ edges 
incident to $a$ and sends them to the referee. The referee computes a maximum matching
in the graph induced by all received edges. As the referee receives a superset of the
edges as described in Algorithm~\ref{alg:ub}, the same approximation factor as in Lemma~\ref{lem:ub}
holds. We hence obtain the following theorem:
\vspace{-0.1cm}
\begin{theorem}
 For every $P \ge 1$, there is a randomized $P$-party simultaneous message protocol for
 maximum matching with expected approximation factor $n^{\alpha}$ and all messages are
 of size $\OrderT(n^{2 - 2\alpha})$.
\end{theorem}

\noindent \textit{Implementation of Algorithm~\ref{alg:ub} as a Dynamic Streaming Algorithm.} 
We employ the technique of $l_0$ sampling in our algorithm \cite{jst11}. For a turnstile stream
that describes a vector $x$, 
a $l_0$-sampler samples uniformly at random from the 
non-zero coordinates of $x$. Similar to Ahn, Guha, and McGregor \cite{agm12}, we employ 
the $l_0$-sampler by Jowhari et al. \cite{jst11}. Their result can be summarized as follows:

\begin{lemma}[\cite{jst11}]
 There exists a turnstile streaming algorithm that performs $l_0$-sampling using space 
 $\Order(\log^2 n \log \delta^{-1})$ with error probability at most $\delta$.
\end{lemma}

In order to implement Algorithm~\ref{alg:ub} in the dynamic streaming setting, for every $a \in A'$, 
we use enough $l_0$-samplers on the sub-stream of incident edges of $a$ in order to guarantee that
with large enough probability, at least $\min\{k, \deg_G(a) \}$ different incident edges of $a$ are sampled.
It can be seen that, for a large enough constant $c$, $c \cdot k \log n$ samplers are enough, 
with probability $1 - \frac{1}{n^{\Theta(c)}}$. We make use of the following lemma whose proof is deferred to the 
appendix.
\begin{lemma} \label{lem:coupon-collector}
 Let $S$ be a finite set, $k$ an integer, and $c$ a large enough constant.  When sampling 
 $c \cdot k \log n$ times from $S$, then with probability $1 - \frac{1}{n^{\Theta(c)}}$, at least 
 $\min\{k, |S| \}$ different elements of $S$ have been sampled.
\end{lemma}
This allows us to conclude with the main theorem of this section.
\begin{theorem} \label{thm:ub-turnstile}
 There exists a one-pass randomized dynamic streaming algorithm for maximum bipartite matching with expected
 approximation ratio $n^\alpha$ using space $\OrderT(n^{2 - 2\alpha})$. 
\end{theorem}

\vspace{-0.6cm}

\bibliography{cc-matching}

\newpage 

\appendix

\section{Auxiliary Lemma}
\begin{lemma} \label{lem:bin-bound}
For positive integers $a,b,c$ so that $c \le a-b$, the following holds:
 \begin{eqnarray*}
  {a-b \choose c} \le  {a \choose c} \cdot \frac{(a-c)^b}{(a-b)^b}.
 \end{eqnarray*}
\end{lemma}
\begin{proof}
 \begin{eqnarray*}
  {a-b \choose c} & = & \frac{(a-b)!}{ (a-b-c)! c!} \le 
  \frac{a!}{ (a-c)! c!} \cdot \frac{(a-c)^b}{(a-b)^b} = {a \choose c} \cdot \frac{(a-c)^b}{(a-b)^b}.  
 \end{eqnarray*} \qed
\end{proof}

\section{Missing Proofs}

\subsection{Missing Proof of Theorem~\ref{thm:rand-lb} }

\textbf{Theorem \ref{thm:rand-lb}.}
 For any $P \le \sqrt{n}$, let $\mathcal{P}_{\mbox{rand}}$ be a $P$-party randomized simultaneous message 
 protocol for maximum matching with error at most $\epsilon$, and all messages are of size at most $s$. Then, $\mathcal{P}_{\mbox{rand}}$ has an approximation factor of $\Omega \left( \left( \frac{Pn}{s} \right)^{\frac{1}{4}} \right)$. 

\begin{proof} 
Let $\mathcal{P}_{\mbox{rand}}$ be a $P$-party randomized simultaneous message protocol for maximum matching 
with error probability at most $\epsilon < 1/2$ and approximation factor $\alpha$. Then, by Yao's lemma, there 
exists a deterministic protocol $\mathcal{P}_{\mbox{det}}$ with approximation ratio $\alpha$, distributional 
error $\epsilon$, and messages of length at most $s$. 

Consider the input distribution as described in Section~\ref{sec:input-dist}, let 
$\mathcal{G}$ denote all possible input graphs, and for a graph $G \in \mathcal{G}$, denote by 
$N_G$ the matching outputted by $\mathcal{P}_{\mbox{det}}$. Furthermore, let 
$\mathcal{G}_{\epsilon} \subseteq \mathcal{G}$ denote those inputs on which $\mathcal{P}_{\mbox{det}}$ errs. 

A maximum matching in $G$ is of size at least $\frac{k (Q+P)}{2}$. As the approximation factor is $\alpha$, we 
have for every $G \in \mathcal{G} \setminus \mathcal{G}_{\epsilon}$: $|N_G| \ge \frac{k(Q+P)}{2 \alpha}$. Hence,
\begin{eqnarray*}
 \Exp_{G \in \mathcal{G}} |N_G| = (1-\epsilon) \cdot \Exp_{G \in \mathcal{G}_{\epsilon}} |N_G| + \epsilon \cdot \Exp_{G \in \mathcal{G} \setminus \mathcal{G}_{\epsilon}} |N_G| \ge (1-\epsilon) \frac{k(Q+P)}{2 \alpha}.
\end{eqnarray*}
From Equation~\ref{eqn:303} from the proof of Theorem~\ref{thm:det-lb}, we obtain
$\Exp_{G \in \mathcal{G}} |N_G| \le 2Qk + P \cdot \Order \left( \frac{\sqrt{sk}}{Q} \right)$,
and hence
\begin{eqnarray*}
 (1-\epsilon) \frac{k(Q+P)}{2\alpha} \le 2Qk + P \cdot \Order \left( \frac{\sqrt{sk}}{Q} \right), \, \mbox{implying} \quad  \alpha = \Omega \left( \frac{k(Q+P)}{Qk + P \cdot \frac{\sqrt{sk}}{Q}} \right) .
\end{eqnarray*}
Note that this term coincides with the term in Inequality~\ref{eqn:005} of the proof of Theorem~\ref{thm:det-lb}.
Optimizing similarly ($k = n/P, Q = \left( \frac{P \sqrt{s}}{\sqrt{k}} \right)^{\frac{1}{2}}$), we obtain
$\alpha = \Omega \left(  \frac{Pn}{s} \right)^{\frac{1}{4}}$. \qed
\end{proof}

\subsection{Missing Proof of Lemma~\ref{lem:coupon-collector} }

\textbf{Lemma \ref{lem:coupon-collector}. }
 Let $S$ be a finite set, $k$ an integer, and $c$ a large enough constant.  When sampling 
 $c \cdot k \log n$ times from $S$, then with probability $1 - \frac{1}{n^{\Theta(c)}}$, at least 
 $\min\{k, |S| \}$ different elements of $S$ have been sampled.

\begin{proof}
We consider the following cases:
\begin{enumerate}
 \item  Suppose that $|S| = k$. Then, we have an instance of the coupon collector's problem. 
 The expected number of times an item $s \in S$ is sampled is $\frac{1}{k} \cdot c \cdot k \log n = c \log n$.
 Then, by a Chernoff bound, the probability that $s$ is not sampled is $\frac{1}{n^{\Theta(c)}}$,
 and using the union bound, the probability that there exists at least one element from $S$ that has not been
 sampled is $\frac{1}{n^{\Theta(c)}}$.
 \item Suppose that $|S| < k$. This case is clearly easier than the case $|S| = k$, as fewer elements 
 have to be sampled (only $|S|$ instead of $k$) and the sampling probability for an element is higher. Therefore, 
 the error probability is smaller than in Case~1.
 \item Suppose now that $|S| > k$. This case is also easier than the case $|S| = k$, since the same
 total number of different samples is required, and the domain from which the samples are chosen from is larger
 ($|S|$ instead of $k$).  Therefore, the error probability is also smaller than in Case~1. 
\end{enumerate} \qed 
\end{proof}

\end{document}